%% file: Quantized_Consensus_over_Static_and_Dynamic_Networks-IEEE.tex
\documentclass[conference, 10pt, twocolumn]{ieeeconf}

\IEEEoverridecommandlockouts
\usepackage[usenames]{color}
\usepackage{enumerate}
\usepackage{url}
\usepackage{subfigure}
\usepackage{amsfonts,mathrsfs}
\usepackage{amssymb,amsmath}
\usepackage{verbatim}
\usepackage{acronym}
\usepackage{graphicx}
\usepackage{amsthm}
\usepackage{ifmtarg}

\def\fskip#1{}

\newtheorem{theorem}{Theorem}

\newtheorem{assumption}{Assumption}

\newtheorem{corollary}{Corollary}

\newtheorem{definition}{Definition}

\newtheorem{lemma}{Lemma}

\newtheorem{remark}{Remark}

\renewcommand{\theenumi}{\arabic{enumi}}

\def\1{{\bf 1}}

\def\R{\mathbb{R}}

\newcommand{\remove}[1]{}

\begin{document}
\title{Convergence Time for Unbiased Quantized Consensus Over Static and Dynamic Networks*}
\author{\authorblockN{Seyed Rasoul Etesami, Tamer Ba\c{s}ar}
  \\ \authorblockA{Coordinated Science Laboratory, University of Illinois at Urbana-Champaign,  Urbana, IL 61801\\
     Email: (etesami1, basar1)@illinois.edu}
\thanks{*Research supported in part by the ``Cognitive \& Algorithmic Decision Making" project grant through the College of Engineering of the University of Illinois, and in part by AFOSR MURI Grant FA 9550-10-1-0573 and NSF grant CCF 11-11342. An earlier version of this paper, on static networks, was presented at the 2013 IEEE Conference on Decision and Control (CDC), and is listed as \cite{etesami2013quantized} in the References section.}
}
\maketitle
\thispagestyle{empty}
\pagestyle{empty} 
\begin{abstract}
In this paper, the question of expected time to convergence is addressed for unbiased quantized consensus on undirected connected graphs, and some strong results are obtained. The paper first provides a tight expression for the expected convergence time of the unbiased quantized
consensus over general but fixed networks. It is shown that the maximum expected
convergence time lies within a constant factor of the maximum hitting time of an appropriate lazy random walk, using the theory of harmonic functions for reversible Markov chains. Following this, and using electric resistance analogy of the reversible Markov chains, the paper provides a tight upper bound for the expected convergence time to consensus based on the parameters of the network. Moreover, the paper identifies a precise order of the maximum expected convergence time for some simple graphs such as line graph and cycle. Finally, the results are extended to bound the expected convergence time of the underlying dynamics in time-varying networks. Modeling such dynamics as the evolution of a time inhomogeneous Markov chain, the paper derives a tight upper bound for expected convergence time of the dynamics using the spectral representation of the networks. This upper bound is significantly better than earlier results for the quantized consensus problem over time-varying graphs.   
\end{abstract}

\smallskip

\section{Introduction}
With the appearance of myriads of online social networks and availability of huge data sets, modeling of the opinion dynamics in a social network has gained a lot of attention in recent years. In a distributed averaging algorithm, agents will exchange their information and update their values based on others' opinions so that eventually they reach the same outcome. Among many problems that arise in such an application is the one of computing the average of agents' initial values. However, in many applications, due to limited memory and energy, agents' states are constrained to be discrete quantities, which leads to finite quantization of states. One of the models that involves such dynamics is the gossip quantized model, where agents' opinions are constrained to be integer valued as it was introduced in \cite{Kashyab} and \cite{franceschelli2011quantized}. The same problem without integer constraints has been studied in many forms; see, for example, \cite{Bertsekas}, \cite{lorenzc}, \cite{distributedaa}. 

\smallskip
In general computing the average and quantized average is known to be an important one in various contexts, such as multi-agent coordination and distributed averaging \cite{lorenzt,nedic2009distributed,olfati2004consensus}, e.g., when a set of robots coordinate in order to move to the same location; information fusion in sensor networks \cite{xiao2007distributed}, e.g., when every sensor in a sensor network has a different measurement of the temperature and the sensors' common goal is to compute the average temperature; task assignment \cite{fanti2013quantized}, e.g., when a group of agents has to reach a consensus on an optimal distribution of tasks, under communication and assignment constraints; decentralized voting \cite{benezit2009interval}, e.g., when nodes initially vote Yes or No and the goal is to find the majority opinion; and load balancing in processor networks, which has various applications in the operation of computer networks \cite{ghosh1994dynamic}.

\smallskip
It is known that, for any initial profile, the quantized dynamics introduced in \cite{Kashyab} will converge to the consensus set with probability 1. It is worth to note that in the case of discrete quantization the final state of the dynamics may or may not be an exact consensus among the agents. Therefore, one can generalize the definition of consensus point to a consensus set which is a set of admissible states such that the opinions of agents in those states are sufficiently close to one another. However, depending on the initial profile and the distribution of choosing the edges at each time instant, the convergence time in expectation may vary. Studies on the behavior of the quantized consensus algorithm based on natural random walks and biased random walks can be found in \cite{robot-krause,etesami2013quantized,shang2013upper}. In fact, simple rule and efficient operation time is one of the main properties of such dynamics.    

\smallskip
In this paper, we consider the quantized consensus algorithm when each edge has equal probability of being chosen at each time step (unbiased). The convergence of such dynamics to the set of quantized consensus points has been shown earlier in \cite{Kashyab}. However, an exact expression for the expected convergence time for these dynamics based on the topology of the network has not yet been given. In this paper we address this issue and provide some tight bounds on the expected convergence time based on some parameters of the underlying graph. The best upper bound for the expected convergence time of the randomized quantized consensus for fixed graphs is known to be $\mathcal{O}(n^3\log n)$ \cite{shang2013upper}, where $n$ is the number of the nodes in the network. However, this bound is given for a randomized quantized consensus with a different protocol than what we consider in this paper. In this paper, we consider unbiased quantized consensus as was introduced in \cite{Kashyab} and provide a tight upper bound of $O(nmD\log(n))$ for the expected convergence time of such dynamics over static networks, where $n$, $m$, and $D$ denote the number of nodes, the number of edges and the diameter of the network, respectively. Therefore, depending on the network topology, this bound can be better or worse than the proposed upper bound in \cite{shang2013upper}. However, an advantage of our analysis for the static networks is that, it provides both tight upper bound and lower bound within a constant factor of each other for the maximum expected convergence time of unbiased quantized consensus. This brings us to a conclusion that one cannot find a better order than what it is obtained in this paper for the maximum expected convergence speed of unbiased quantized consensus dynamics. Therefore, for faster convergence speed for the quantized consensus problem one must consider protocols other than the unbiased protocol. We take it even one step further and study the expected convergence time of such dynamics over time varying networks. In particular, we prove an upper bound for the expected convergence time in the case of time-varying networks, which is significantly tighter than the protocol given in \cite{zhu2011convergence} by a factor of at least $n^4$, where $n$ is the number of agents involved in the dynamics.      

\smallskip
The paper is organized as follows: In Section~\ref{sec:HKdyn}, we review the unbiased quantized consensus dynamics and some of its basic properties. Further, we identify some connections between the expected convergence time and random walks on the graphs. In Section~\ref{sec:mainresults}, we obtain our main results related to fixed networks, providing expected convergence time for these dynamics. Moreover, we include some simulation results for simple but fixed graphs. In Section \ref{sec:Time-Varying} we extend our results to time-varying networks, and establish a polynomial upper bound for the expected convergence time of the quantized consensus. We conclude the paper in Section~\ref{sec:conclusion}. To ease exposition of our main results in the main body of the paper, we relegate their proofs as well as some important relevant results from the theory of Markov chains and related topics to an appendix (Appendix \ref{ap-static-preliminary-proofs}). In particular we provide in Appendix I an overview of some relevant results related to comparison between random walks and electric networks. Finally, in Appendix \ref{ap-static-alternative-proofs}, we present an alternative approach to study the maximum expected convergence time of unbiased quantized consensus for the case of simple static graphs such as line graph and cycle.\\
           
\textbf{Notations}: 
For a vector $v\in \R^n$, we let $v_i$ to be the $i$th entry of $v$, and $v^T$ be the transpose of $v$. We let $\bold{1}$ be a column vector with all entries equal to 1. We say that $v$ is stochastic if $v_i\geq 0$ for all $i\in \{\1,2,\ldots, n\}$ and $\sum_{i=1}^nv_i=1$. Similarly, we say that a matrix $A$ is stochastic if each row of $A$ is stochastic. If both $A$ and $A^T$ are stochastic, we say that
$A$ is doubly stochastic. For a random walk $\mathcal{Z}$ with transition probability matrix $P$, we let the random variable $\tau^a_z$ be the first time that the random walk initiated at $a$ hits the state $z$. Also, we let $H_{\mathcal{Z}}(a,z)$ denote the expected time that the random walk $\mathcal{Z}$ initiated at $a$ hits $z$ for the first time. We take $G_{\tau^a_z}(x)$ to be the expected number of visits to $x$ before $\tau^a_z$. For an undirected graph $\mathcal{G}=(V,\mathcal{E})$, we let $N(x)$ be the set of neighbors of $x$. We also let $\mathcal{G}\times \mathcal{G}=(V\times V, \mathcal{E}')$ be the Cartesian product of $\mathcal{G}$, i.e. $\big((x,y),(r,s)\big)\in \mathcal{E}'$ if and only if $x=r, s\in N(y)$ or $y=s, r\in N(x)$. For a graph $\mathcal{G}$, we let $\mathcal{A}_{\mathcal{G}}$ be its adjacency matrix and $\mathcal{D}_{\mathcal{G}}$ be a diagonal matrix whose diagonal entries are equal to the degree of the nodes in the graph. Moreover, we denote the Laplacian of this graph by $\mathcal{L}_{\mathcal{G}}=\mathcal{D}_{\mathcal{G}}-\mathcal{A}_{\mathcal{G}}$. We let  $\mathcal{R}(x\leftrightarrow y)$ be the effective resistance between two nodes $x$ and $y$ in an electric network when every edge has resistance equal to 1. We denote by $\mathbb{V}(\cdot)$ the voltage function over an electric circuit for two distinguished nodes $a$ and $z$ with $\mathbb{V}(a)=1$ and $\mathbb{V}(z)=0$, and let $\mathbb{V}_{xy}=\mathbb{V}(x)-\mathbb{V}(y)$. Finally, we define the distribution of a vector $v$ as the list $\{(v_1, n_1),(v_2, n_2),\ldots \}$ in which $n_i$ is the number of entries of $v$ which have value $v_i$.

\smallskip
\section{Unbiased Quantized Consensus }\label{sec:HKdyn}

 In this section, we describe the discrete-time quantized consensus model for a fixed network as introduced in \cite{Kashyab} and postpone the analysis of time-varying networks to Section \ref{sec:Time-Varying}.

\smallskip
\subsection{Model}

\smallskip
\begin{itemize}
\item There is a set of $n$ agents, $V=\{1,2,\ldots,n\}$, which are connected on some undirected graph $\mathcal{G}(V,\mathcal{E})$.
\item Each agent has an initial value $x_i(0)$, which is a positive integer.
\item At each time instant $t=1,2, \ldots$, one edge is chosen uniformly at random among the set of all the edges $\mathcal{E}$,  
and the incident nodes on the sides of this edge (let us call them $i$ and $j$) update their values according to:
\begin{align}\label{eq:Quantize-rule}
x_i(t+1) = \begin{cases} x_i(t)-1, & \mbox{if} \ x_i(t)>x_j(t) \\ 
x_i(t)+1, & \mbox{if}  \ x_i(t)<x_j(t) \\
x_i(t), & \mbox{if}  \ x_i(t)=x_j(t), 
\end{cases}
\end{align}
and the same holds for agent $j$. We refer to $x_i(t)$ as the \textit{opinion of agent $i$} at time $t$ and $x(t)$ as the \textit{opinion profile} at time $t$.
\end{itemize}

\smallskip
\subsection{Relevant results on maximun expected convergence time}\label{sec:relevant}

\smallskip
It has been shown in \cite{Kashyab} that at each time instant $t$, the Lyapunov function defined by:
\begin{align}\label{eq:Lyapunov} 
V(x(t))=\sum_{\ell=1}^{n}\Big(x_{\ell}(t)-\frac{\sum_{k=1}^{n}x_k(0)}{n}\Big)^2
\end{align}
will decrease by at least 2 if a nontrivial update occurs at $t$. By nontrivial update we mean that the values of the incident nodes of the chosen edge $(i,j)$ at time instant $t$ differ by at least 2, i.e. $|x_i(t)-x_j(t)|\ge 2$. Therefore, to bound the expected convergence time for the above algorithm, one can think of the maximum expected time it takes for a nontrivial update to occur. Also, using the Lyapunov function given in \eqref{eq:Lyapunov}, one can see that the number of nontrivial updates is at most of the order of $O(n)$. In fact it has been shown in \cite{Kashyab} that $\frac{(L-l)^2}{8}n$ nontrivial updates are enough to guarantee the termination of the dynamics, where $L$ and $l$ denote, respectively, the maximum and minimum opinions at the beginning of the process. Therefore, the problem reduces to that of finding the maximum expected time it takes for a nontrivial update to take place, which we denote by $\bar{T}(\mathcal{G})$. In this case the expected convergence time of the quantized consensus will be upper bounded by $\frac{(L-l)^2}{8}n\bar{T}(\mathcal{G})$. In this paper we are interested in computing $\bar{T}(\mathcal{G})$. Let $T_1(x(0))$ be a random variable denoting the first nontrivial averaging when the initial profile is $x(0)$. Through some manipulation, it is not hard to see that $\bar{T}(\mathcal{G})=\max_{x(0)\in \mathcal{X(\mathcal{G})}}{T_1(x(0))}$ where $\mathcal{X(\mathcal{G})}=\{x| \mbox{distribution of x is} \ \{(0,1), (1, n-2), (2,1)\}\}$, i.e., $\mathcal{X(\mathcal{G})}$ is the set of all the vectors of
size $n$ in which one entry is 0, one entry is 2, and the remaining entries are 1. Therefore, the main issue is to find an expression for $\bar{T}(\mathcal{G})=\max_{x(0)\in \mathcal{X(\mathcal{G})}}{T_1(x(0))}$. 

\smallskip
\subsection{Random walk interpretation}

\smallskip
In the above setting, we now assume that all the agents on the graph $\mathcal{G}$ have value 1 except two of them which are 0 and 2. At each time instant $t$, one edge will be selected with equal probability $\frac{1}{m}$ where $m$ is the number of edges, and the incident nodes update their values based on \eqref{eq:Quantize-rule}. Therefore, we can interpret this problem in an alternative way. Consider two walkers, let us call them 0 and 2, who start a random walk on the vertices of the graph $\mathcal{G}$ whenever the selected edge is incident to at least one of them. To see this more clearly, let us consider a network of $n$ nodes such that all of the nodes have value 1 except two of them which have values 0 and 2. Therefore, in the next update of the protocol \eqref{eq:Quantize-rule}, either the selected edge is incident to neither of the nodes with values 0 or 2, in which case there will not be any change, or the selected edge is incident to the node with value 0 and one of the nodes with value 1 (similarly to node with value 2 and one of the nodes with value 1). In this case 0 and 1 on the sides of the selected edge will be swapped (analogously, 2 and 1 on the sides of the selected edge will be swapped). This can be viewed as a random walk that the nodes with values 0 and 2 take to their next positions. Therefore, $\bar{T}(\mathcal{G})$ is equal to the maximum of the expected time it takes for these two walkers to meet. Based on this interpretation we have the following definition.
\smallskip
\begin{definition}\label{def:original-process-static}
Denoting the current locations of the walkers by $x$ and $y$, if the selected edge at the next time instant is incident to one of the walkers, e.g., $\{x,x_i\}$ for some $x_i\in N(x)$, we will move that walker from node $x$ to node $x_i$, otherwise the walkers will not change their positions. We refer to such a random walk process as the \textit{original} process.
\end{definition} 

\smallskip
One important fact is that both of the walkers in the original process have the same source of randomness, which selects an edge at each time instant. Therefore, these random walks are jointly correlated. In fact, in the Section \ref{sec:mainresults}, we find an explicit form for $\bar{T}(\mathcal{G})$ for a general fixed network.    

\bigskip
\section{Main Results for static networks}\label{sec:mainresults}

\smallskip
Let us consider the original joint process for two walkers on a finite graph, meaning that either of the walkers will move depending on whether the selected edge is incident to it. In order to compute $\bar{T}(\mathcal{G})$, we introduce another process, called \textit{virtual} process, to facilitate
our analysis. 

\smallskip
We let the virtual process and the original process be the same until the time when the walkers become each other's neighbors, i.e. $x\in N(y)$, for some $x,y\in V$. At this time we count the connecting edge in the virtual process twice in our edge probability distribution. Moreover, we denote the meeting time function of the virtual process by $M(x,y)$ for every two initial states $x$ and $y$. Note that, in the virtual process, as long as the walkers are not each other's neighbors, one edge of the network is picked, with probability $\frac{1}{m}$ and the movement of the walkers will be precisely as in the original process. But if the walkers are each other's neighbors, i.e. $x\in N(y)$, the edge selection probability slightly changes. In this case, the probability of selecting an edge $\mathbb{P}(e), e\in \mathcal{E}(\mathcal{G})$ will be as follows:
\begin{align}\nonumber
\!\mathbb{P}(e) \!=\! \begin{cases} \ \frac{2}{m},\!\!\!\!\!\! & \mbox{if} \ e=\{x,y\} \\ 
\frac{1}{m},\!\!\!\!\!\! & \mbox{if} \ \mbox{$e$ is incident to either $x$ or $y$}   \\
\!\!\frac{m-d(x)-d(y)}{m(m+1-d(x)-d(y))},\!\!\!\!\! & \mbox{if} \ \mbox{$e$ is incident to neither $x$ nor $y$}, 
\end{cases}
\end{align}
and the walkers move depending on whether the selected edge is incident to them or not.

\smallskip
\begin{remark}\label{rem:star-graph}
In order to have valid transition probabilities for the virtual process, we must have $d(x)+d(y)\leq m, \forall (x,y)\in \mathcal{E}(\mathcal{G})$. This condition naturally holds for all connected graphs except the star graph and double-star graph, i.e., two star graphs whose centers are connected to each other (Figure \ref{fig:double_star}). In these cases we have $\max\{d(x)+d(y)|\ (x,y)\in \mathcal{E}(\mathcal{G})\}=m+1$. However, the maximum expected meeting time of these two special cases can be computed precisely and directly without using a virtual process (Lemma \ref{lemm:special-case-star}). Therefore, henceforth we assume that $d(x)+d(y)\leq m, \forall (x,y)\in \mathcal{E}(\mathcal{G})$.      
\end{remark}

\smallskip
Given a graph $\mathcal{G}$, we choose an edge at random and uniformly among all the set of edges at each time instance. Let $\mathcal{Z}$ be the lazy random walk which is generated based on marginal distribution of this setting. In other words, the walker will move towards one of his neighbors with equal probability if he is located at one of the incident vertices. Hence, one can interpret $\mathcal{Z}$ as a random walk which is generated based on marginal distribution of only one of the walkers in the original process. It is not hard to see that $\mathcal{Z}$ has the following transition probabilities:
\begin{align}\label{eq:transition-z}
P_{\mathcal{Z}}(x,y) = \begin{cases} 1-\frac{d(x)}{m}, & \mbox{if} \ y=x \\ 
\frac{1}{m}, & \mbox{if} \  y \in N(x) \\
0, & \mbox{else}. 
\end{cases}
\end{align}
Since the above transition matrix is doubly stochastic, $\pi=(\frac{1}{n}, \frac{1}{n}, \ldots, \frac{1}{n})^T$ is its stationary distribution. This results in $\pi_iP(i,j)=\pi_jP(j,i)\ \forall i,j\in 1,\ldots, n$, and hence $\mathcal{Z}$ is a reversible Markov chain.  

\smallskip
We already know \cite{Fundamental} that every reversible Markov chain has a hidden vertex $w$ such that the hitting time from $w$ to every state is less than or equal to the hitting time from that particular node to state $w$, i.e, $w$ is a hidden vertex for $\mathcal{Z}$ if  $H_{\mathcal{Z}}(w,x)\leq H_{\mathcal{Z}}(x,w) \ \forall x$.

\smallskip 
\begin{definition}
Assume that $w$ is a hidden vertex for the reversible Markov chain $\mathcal{Z}$. As in \cite{Fundamental} and \cite{shang2013upper}, we define the potential function $\Phi(\cdot ,\cdot):V\times V\rightarrow R$ to be 
\begin{align}\nonumber
\Phi(x,y)=H_{\mathcal{Z}}(x,y)+H_{\mathcal{Z}}(y,w)-H_{\mathcal{Z}}(w,y).
\end{align}
\end{definition}

\begin{definition}
A function $h:\Omega\rightarrow \mathbb{R}$ is called harmonic at a vertex $x\in \Omega$ for a Markov chain with transition probability matrix $P$ if $h(x)=\sum_{y\in \Omega}P(x,y)h(y)$.
\end{definition}

We are now in a position to start our analysis. First, we briefly describe the stages that we will go through toward proving the result for general static networks. As discussed earlier, $\bar{T}(\mathcal{G})$ is equal to the maximum expected meeting time of the original process. Since, due to the coupling between the random walks, computing the expected meeting time of the original process is difficult, we approximate the original process with the virtual process which are almost always the same, except when the walkers are each other's neighbors. However, the virtual process is itself a jointly correlated random walk. Therefore, to characterize its expected meeting time function, i.e., $M(x,y)$, we will show that $M(x,y)$ follows almost the same recursion formula as $\Phi(x,y)$. This allows us to construct a harmonic function (Lemma \ref{lemm:harmonic-main-lemma}) using $\Phi(x,y)$ and $M(x,y)$. We show that such a harmonic function is zero at some boundary point, and hence, must be identical to zero. This allows us to characterize $M(x,y)$ based on $\Phi(x,y)$ (Theorem \ref{thm:main}). Furthermore, since $\Phi(x,y)$ is a function of the expected hitting time of a single lazy random walk $\mathcal{Z}$, we can find an expression for $M(x,y)$ based on only the expected hitting time functions of the lazy random walk $\mathcal{Z}$. Moreover, since such an expression does not involve any coupling term, it is easy to compute it for different networks. Finally, we show in Theorem \ref{thm:second-main} that the expected meeting time function of the virtual process $M(x,y)$, and that in the original process lie within a constant factor of each other. This establishes our tight bound for $\bar{T}(\mathcal{G})$. We now start the stages of our proof, with the following lemma.    

\smallskip
\begin{lemma}\label{lemm:harmonic-main-lemma}
A function $f:V\times V\rightarrow \mathbb{R}$ defined by $f(x,y)=\frac{1}{2}\Phi(x,y)-M(x,y)$ is harmonic for the simple random walk on $\mathcal{G}\times \mathcal{G}$, i.e. 
\begin{align}\nonumber
f(x,y)=\sum_{(r,s)\in V\times V}\mathcal{Q}\big((x,y),(r,s)\big)f(r,s),
\end{align} 
where $\mathcal{Q}$ is the transition matrix of the simple random walk on $\mathcal{G}\times \mathcal{G}$, i.e.
\begin{align}\nonumber
\mathcal{Q}\big((x,y),(r,s)\big) = \begin{cases} \frac{1}{d(x)+d(y)}, & \mbox{if} \ (r,s) \in N_{\mathcal{G}\times \mathcal{G}}(x,y) \\ 
0, & \mbox{else}. 
\end{cases}
\end{align}
\end{lemma}
\vspace{0.3cm}
\begin{proof} 
By the transitivity property of reversible Markov chains, we note that $\Phi(x,y)$ is symmetric, i.e. for any hidden vertex $w$, 
\begin{align}\nonumber
\Phi(x,y)&=H_{\mathcal{Z}}(x,y)+H_{\mathcal{Z}}(y,w)-H_{\mathcal{Z}}(w,y)\cr &=H_{\mathcal{Z}}(y,x)+H_{\mathcal{Z}}(x,w)-H_{\mathcal{Z}}(w,x)=\Phi(y,x).
\end{align}
Therefore, we can write:
\begin{align}\label{Lemma:Phi-one}
\Phi(x,y)&=\frac{d(x)}{d(x)+d(y)}\big[H_{\mathcal{Z}}(x,y)+H_{\mathcal{Z}}(y,w)-H_{\mathcal{Z}}(w,y)\big]\cr &+\frac{d(y)}{d(x)+d(y)}\big[H_{\mathcal{Z}}(y,x)+H_{\mathcal{Z}}(x,w)-H_{\mathcal{Z}}(w,x)\big]\cr 
&=\frac{d(x)}{d(x)+d(y)}\big(H_{\mathcal{Z}}(y,w)-H_{\mathcal{Z}}(w,y)\big)\cr &+\frac{d(y)}{d(x)+d(y)}\big(H_{\mathcal{Z}}(x,w)-H_{\mathcal{Z}}(w,x)\big)\cr &+\big[\frac{d(x)}{d(x)+d(y)}H_{\mathcal{Z}}(x,y)+\frac{d(y)}{d(x)+d(y)}H_{\mathcal{Z}}(y,x)\big].
\end{align}
Also, by expanding $H_{\mathcal{Z}}(x,y)$ by one step, we get:
\begin{align}\label{Lemma:Phi-Hxy}
H_{\mathcal{Z}}(x,y)=\frac{m}{d(x)}+\frac{1}{d(x)}\sum_{j\in N(x)}H_{\mathcal{Z}}(j,y),
\end{align}
and similarly by switching $x$ and $y$ we have:
\begin{align}\label{Lemma:Phi-Hyx}
H_{\mathcal{Z}}(y,x)=\frac{m}{d(y)}+\frac{1}{d(y)}\sum_{j\in N(y)}H_{\mathcal{Z}}(j,x).
\end{align}
Using \eqref{Lemma:Phi-Hxy} and \eqref{Lemma:Phi-Hyx} in \eqref{Lemma:Phi-one}, we get

\begin{align}\label{Lemma:Phi-two}
\Phi(x,y)&\!=\!\frac{2m}{d(x)+d(y)}\!+\!\frac{d(x)}{d(x)\!+\!d(y)}\big(H_{\mathcal{Z}}(y,w)\!-\!H_{\mathcal{Z}}(w,y)\big)\cr
&\!+\!\frac{d(y)}{d(x)\!+\!d(y)}\big(H_{\mathcal{Z}}(x,w)\!-\!H_{\mathcal{Z}}(w,x)\big)\cr
&\!+\!\frac{1}{d(x)\!+\!d(y)}\Big(\!\!\sum_{j\in N(x)}\!\!H_{\mathcal{Z}}(j,y)\!+\!\!\sum_{j\in N(y)}\!\!H_{\mathcal{Z}}(j,x)\Big).
\end{align}

Also, from the definition of $\Phi(\cdot, \cdot)$, we have $\Phi(j,y)=H_{\mathcal{Z}}(j,y)+H_{\mathcal{Z}}(y,w)-H_{\mathcal{Z}}(w,y), \ \forall j\in N(x)$. By taking summation over all $j\in N(x)$ and multiplying by the factor $\frac{1}{d(x)+d(y)}$, we arrive at
\begin{align}\label{Lemma:Phi-three} 
\frac{1}{d(x)+d(y)}&\sum_{j\in N(x)}\!\!\Phi(j,y)=\frac{1}{d(x)+d(y)}\sum_{j\in N(x)}\!\!H_{\mathcal{Z}}(j,y)\cr&+\frac{d(x)}{d(x)+d(y)}\big[H_{\mathcal{Z}}(y,w)-H_{\mathcal{Z}}(w,y)\big].
\end{align} 
By the same argument, and since $\Phi(x,j)=H_{\mathcal{Z}}(j,x)+H_{\mathcal{Z}}(x,w)-H_{\mathcal{Z}}(w,x)$, we have
\begin{align}\label{Lemma:Phi-four} 
\frac{1}{d(x)+d(y)}&\sum_{j\in N(x)}\!\!\Phi(j,y)=\frac{1}{d(x)+d(y)}\sum_{j\in N(y)}\!\!H_{\mathcal{Z}}(j,x)\cr&+\frac{d(y)}{d(x)+d(y)}\big[H_{\mathcal{Z}}(x,w)-H_{\mathcal{Z}}(w,x)\big],
\end{align} 
Substituting \eqref{Lemma:Phi-four} and \eqref{Lemma:Phi-three} in \eqref{Lemma:Phi-two} gives us
\begin{align}\label{Lemma:phi-harmonic}
\!\!\!\!\!\!\!\!\!\Phi(x,y)&\!=\!\frac{2m}{d(x)+d(y)}\cr 
&\!+\!\frac{1}{d(x)+d(y)}\Big(\!\!\!\sum_{j\in N(x)}\!\!\Phi(j,y)\!+\!\!\!\!\sum_{j\in N(y)}\!\!\Phi(j,x)\Big).
\end{align}
  
On the other hand, we note that regardless of whether $y \notin N(x)$ or $y \in N(x)$, the meeting time of the virtual process is equal to 
\begin{align}\nonumber
M(x,y)&\!=\!(1\!-\!\frac{d(x)+d(y)}{m})\big(1\!+\!M(x,y)\big)\cr&+\!\!\!\sum_{j\in N(x)}\!\!\frac{1}{m}\big(1\!+\!M(j,y)\big)\!+\!\!\!\sum_{j\in N(y)}\!\!\frac{1}{m}\big(1\!+\!M(j,x)\big)
\end{align}
from which by simplifying and rearranging the terms we get
\begin{align}\label{eq:M-harmoney}
\!\!\!\!\!\!\!M(x,y)\!&=\!\frac{m}{d(x)\!+\!d(y)}\cr 
&\!+\!\frac{1}{d(x)\!+\!d(y)}\Big(\!\!\!\sum_{j\in N(x)}\!\!\!M(j,y)\!+\!\!\!\!\sum_{j\in N(y)}\!\!\!M(j,x)\!\Big).
\end{align}
Let $S(x,y)=\frac{\Phi(x,y)}{2}$. From \eqref{Lemma:phi-harmonic} it is not hard to see that
\begin{align}\label{eq:S-harmoney} 
\!\!\!\!\!\!\!S(x,y)\!&=\!\frac{m}{d(x)\!+\!d(y)}\cr 
&\!+\!\frac{1}{d(x)\!+\!d(y)}\Big(\!\!\!\sum_{j\in N(x)}\!\!S(j,y)\!+\!\!\!\sum_{j\in N(y)}\!\!S(j,x)\!\Big).
\end{align}
We consider the simple random walk $\mathcal{Q}$ on the Cartesian product graph $\mathcal{G}\times \mathcal{G}$. The cover time and hitting time of such graphs have been extensively studied in \cite{productcover}, \cite{Covertime} and \cite{Explore-fast-graph}. We show that the function $f(x,y)=S(x,y)-M(x,y)$ is harmonic on $\mathcal{G}\times \mathcal{G}$ for the transition matrix $\mathcal{Q}$. In fact,
\begin{align}\nonumber
f(x,y)&=\frac{1}{d(x)+d(y)}\sum_{j\in N(x)}\big(S(j,y)-M(j,y)\big)\cr&+\frac{1}{d(x)+d(y)}\sum_{j\in N(y)}\big(S(j,x)-M(j,x)\big)\cr &=\frac{1}{d(x)+d(y)}\Big(\sum_{j\in N(x)}f(j,y)+\sum_{j\in N(y)}f(x,j)\Big)\cr&=\sum_{(r,s)\in V\times V}\mathcal{Q}\big((x,y),(r,s)\big)f(r,s).
\end{align}
This completes the proof.
\end{proof}

\smallskip 
Now we are ready to characterize the expected meeting time of the virtual process based on the expected hitting times of the single lazy random walk $\mathcal{Z}$ and effective resistances of an appropriate network. 

\smallskip
\begin{theorem}\label{thm:main}
The expected meeting time of the virtual process initiated from $x,y$ is equal to $\frac{1}{2}\Phi(x,y)$, i.e.,
\begin{align}\nonumber
M(x,y)=\frac{1}{2}\Big[H_{\mathcal{Z}}(x,y)\!+\!H_{\mathcal{Z}}(y,w)\!-\!H_{\mathcal{Z}}(w,y)\Big].
\end{align}
\end{theorem}

\smallskip
\begin{proof}
Let $g:V\times V\rightarrow \mathbb{R}$ be the zero function, i.e., $g\equiv 0$. Clearly, $g$ is a harmonic function over $\mathcal{G}\times \mathcal{G}$. On the other hand, we have
\begin{align}\nonumber
f(w,w)&\!=\!S(w,w)\!-\!M(w,w)\!=\!S(w,w)\cr &\!=\!\frac{1}{2}\Big(H(w,w)\!+\!H(w,w)\!-\!H(w,w)\!\Big)\!=\!0\!=\!g(w,w).
\end{align} 
Since $f$ and $g$ are both harmonic functions for the transition matrix $\mathcal{Q}$ and also they have the same value at the node $(w,w)$, using Lemma \ref{Lemma:harmonic} they must be equal. Thus $f\equiv 0$, which shows that $M(x,y)=\frac{1}{2}\Phi(x,y), \forall x,y$.
\end{proof} 

\bigskip
\begin{theorem}\label{thm:second-main}
Consider a network $\mathcal{G}=(V,\mathcal{E})$. Then, 
\begin{align}\nonumber
\max_{x,y}M(x,y) \leq \bar{T}(\mathcal{G})\leq 2\max_{x,y}M(x,y).
\end{align}      
\end{theorem}
\begin{proof}
Initiating from arbitrary nodes $x$ and $y$, we note that both the virtual process and the original process follow the same joint distribution until they are each other's neighbor. However, when the walkers are each other's neighbor, with higher probability they are going to meet in the virtual process than in the original process.  Therefore, $M(x,y)\leq\bar{T}(\mathcal{G})$ for all $x, y$. Since by definition $\bar{T}(\mathcal{G})$ is independent of the initial states of walkers, $\max_{x,y}M(x,y) \leq \bar{T}(\mathcal{G})$. For the upper bound, we use the same argument as in \cite{shang2013upper}. Again, as mentioned earlier, the virtual process and the original process remain the same until the two walkers become each other's neighbors, i.e. for some $x,y$ with $x\in N(y)$. At this time, the probability that two walkers meet in the next transition in the original process is $\frac{1}{m}$, while this probability for the virtual process is $\frac{2}{m}$. Since the former is half of the latter, this immediately implies that the meeting time of the former is within a constant factor of the latter. In fact, at each time that the walkers in the virtual process meet, with probability $\frac{1}{2}$ the walkers in the original process meet as well. However, if the walkers in the virtual process meet, but they do not meet in the original process (which happens with probability $\frac{1}{2}$) then, in the original process we may assume that the positions of the walkers have not changed, and in the virtual process we may assume that they just switch their positions (from $x,y$ to $y,x$). Therefore, in this case we may assume that a new original process which is followed by its corresponding virtual process has just been initiated from nodes $x$ and $y$. Since each of these collisions of walkers during different time intervals happens independently and with probability $\frac{1}{2}$, we can write 
\begin{align}\nonumber
\bar{T}(\mathcal{G})\leq \sum_{k=1}^{\infty}(\frac{1}{2})^kk\max_{x,y}M(x,y)=2\max_{x,y}M(x,y), 
\end{align} 
where in the above summation the term $(\frac{1}{2})^k$ corresponds to the probability that the walkers in the virtual process meet $k$ times while they do not meet in the original process, and the term $k\max_{x,y}M(x,y)$ is an upper bound for the expected time that the walkers in the virtual process meet $k$ times.   
\end{proof}

Next, we will proceed by computing $H_{\mathcal{Z}}(x,y)$. Note that if at the time instant $t$ the walker is at the node $\mathcal{Z}(t)$, then the probability of staying in that state is $1-\frac{d(\mathcal{Z}(t))}{m}$. Because of Markov property of the random walk, the probability of moving out from each state follows the geometric distribution. Therefore, the expected time that the walker waits in state $\mathcal{Z}(t)$ is $\frac{m}{d(\mathcal{Z}(t))}$. However, when the walker is moving to the next state, he will see all of his neighbors with the same probability. Therefore, we can split the expected hitting time of the random walk $\mathcal{Z}(t)$ namely $H_{\mathcal{Z}}(x,y)$ to summation of two parts: 
\begin{itemize}
\item the expected time that the walker spends in each state $\mathcal{Z}(t)$, which is $\frac{m}{d(\mathcal{Z}(t))}$, 
\item the expected hitting time of a simple random walk $\mathcal{Z'}(t)$, namely $H(x,y)$.  
\end{itemize} 
Therefore, we have:
\begin{align}\label{eq:partition-lazy}
H_{\mathcal{Z}}(x,y)&\!=\!\!\!
\sum_{t=1}^{H(x,y)}\!\!\big(1\!+\!\frac{m}{d(\mathcal{Z'}(t))}\big)\!=\!H(x,y)\!+\!m\!\!\sum_{t=1}^{H(x,y)}\!\frac{1}{d(\mathcal{Z'}(t))}\cr 
&\!=\!H(x,y)\!+\!m\sum_{i}\frac{G_{\tau^x_y}(i)}{d(i)}\cr 
&\!=\!\!\sum_{i}\frac{m\!+\!d(i)}{2}\big[\mathcal{R}\big(x\leftrightarrow y\big)\!+\!\mathcal{R}\big(y \leftrightarrow i\big)\!-\!\mathcal{R}\big(x\leftrightarrow i\big)\big]\cr 
&\!\leq\! m\sum_{i}\big[\mathcal{R}\big(x\leftrightarrow y\big)\!+\!\mathcal{R}\big(y \leftrightarrow i\big)\!-\!\mathcal{R}\big(x\leftrightarrow i\big)\big]
\end{align}  
where the last equality is due to Lemma \ref{lemm:triangleineq} and relation \eqref{eq:hitting}.

\smallskip
We are now ready to state our upper and lower bounds for the expected convergence time of the unbiased quantized consensus over static networks.   

\bigskip
\begin{theorem}\label{thm:end-main}
Consider a connected network $\mathcal{G}$ with $n$ nodes and $m$ edges and diameter $D$. Then, for unbiased quantized consensus, we have
\begin{align}\nonumber
\frac{1}{2}H_{\mathcal{Z}}\leq \bar{T}(\mathcal{G})\leq H_{\mathcal{Z}}\leq 2nmD,
\end{align}
where, $H_{\mathcal{Z}}=\max_{x,y}H_{\mathcal{Z}}(x,y)$.
\end{theorem}
\begin{proof}
Since the transition matrix given in \eqref{eq:transition-z} is symmetric, hence the random walk $\mathcal{Z}$ is a symmetric random walk. Therefore, every path has the same probability backward and forward and we have $H_{\mathcal{Z}}(y,w)=H_{\mathcal{Z}}(w,y), \forall y$. Applying this equality in the expression for $M(x,y)$ in Theorem \ref{thm:main} gives us
\begin{align}\label{eq:upper-lower}
M(x,y)=\frac{1}{2}H_{\mathcal{Z}}(x,y), \forall x,y.
\end{align}
Using \eqref{eq:upper-lower} and Theorem \ref{thm:second-main}, and the definition of $H_{\mathcal{Z}}$, we can see that $\frac{1}{2}H_{\mathcal{Z}} \leq \bar{T}(\mathcal{G})\leq H_{\mathcal{Z}}$. Furthermore, by relation \eqref{eq:partition-lazy}, we get 
\begin{align}\label{eq:upper-max-hitting-time}
H_{\mathcal{Z}}&\leq \max_{x,y}\Big\{ m\sum_{i}\big[\mathcal{R}\big(x\leftrightarrow y\big)+\mathcal{R}\big(y \leftrightarrow i\big)-\mathcal{R}\big(x\leftrightarrow i\big)\big]\Big\}\cr 
&\leq 2nmD,
\end{align}
where, in the second inequality we have used the fact that the effective resistance between any two nodes cannot exceed the length of the shortest path between those nodes \cite{Markov-Book} which is upper bounded by the diameter of $\mathcal{G}$. 
\end{proof}

\smallskip
\begin{remark}
As discussed earlier in Section \ref{sec:relevant}, the total number of nontrivial
averagings is at most of the order of $\mathcal{O}(n)$. This, in view
of the theorem \ref{thm:end-main} and Corollary 4 in \cite{shang2013upper}, shows that the maximum expected convergence time of unbiased quantized consensus over static network is at most $\mathcal{O}(nmD\log(n))$.
\end{remark}

\smallskip
\subsection{Simulation Results}\label{sec:Simulation}
In this section we present some simulation results to provide a comparison between the maximum expected meeting time $\bar{T}(\mathcal{G})$ and the proposed upper and lower bounds given above. We consider four different types of graphs with $n$ nodes: line graph, star graph, lollipop graph, and semi-regular graph. In lollipop graph each of its side clusters has $[\frac{n}{4}]$ nodes and they are connected with a single path. Also, for the semi-regular graph we consider a graph with $n$ nodes arranged around a circle, such that each node is linked to its next four nodes when we move clockwise around the circle. In Figure \ref{fig:all-graphs}, the ratio $\frac{\bar{T}(\mathcal{G})}{mnD}$ is depicted for each graph. As it can be seen, this ratio for the line graph converges asymptotically to a constant as $n$ goes to infinity. This is the same as what we would have expected from Corollary \ref{corr:equality-line-graph}. Moreover, using the transition probabilities given in \eqref{eq:transition-z} for the star graph, a simple calculation shows that $H_{\mathcal{Z}}=\frac{n(n-1)}{2}$, and hence from Theorem \ref{thm:end-main} we get $\frac{n(n-1)}{4n^2}\leq\frac{\bar{T}(\mathcal{G})}{mnD}\leq \frac{n(n-1)}{n^2}$, which is consistent with the ratio given in Figure \ref{fig:all-graphs}. Finally, for lollipop and semi-regular graphs, although the ratio $\frac{\bar{T}(\mathcal{G})}{mnD}$ is oscillating, it is clearly bounded from above by 1, which confirms the upper bound provided in Theorem \ref{thm:end-main}.                
 
\smallskip
\begin{figure}[htb]
\vspace{-3cm}
\begin{center}
\vspace{-0.75cm}
\hspace{1cm}
\includegraphics[totalheight=.4\textheight,
width=.6\textwidth,viewport=0 0 500 500]{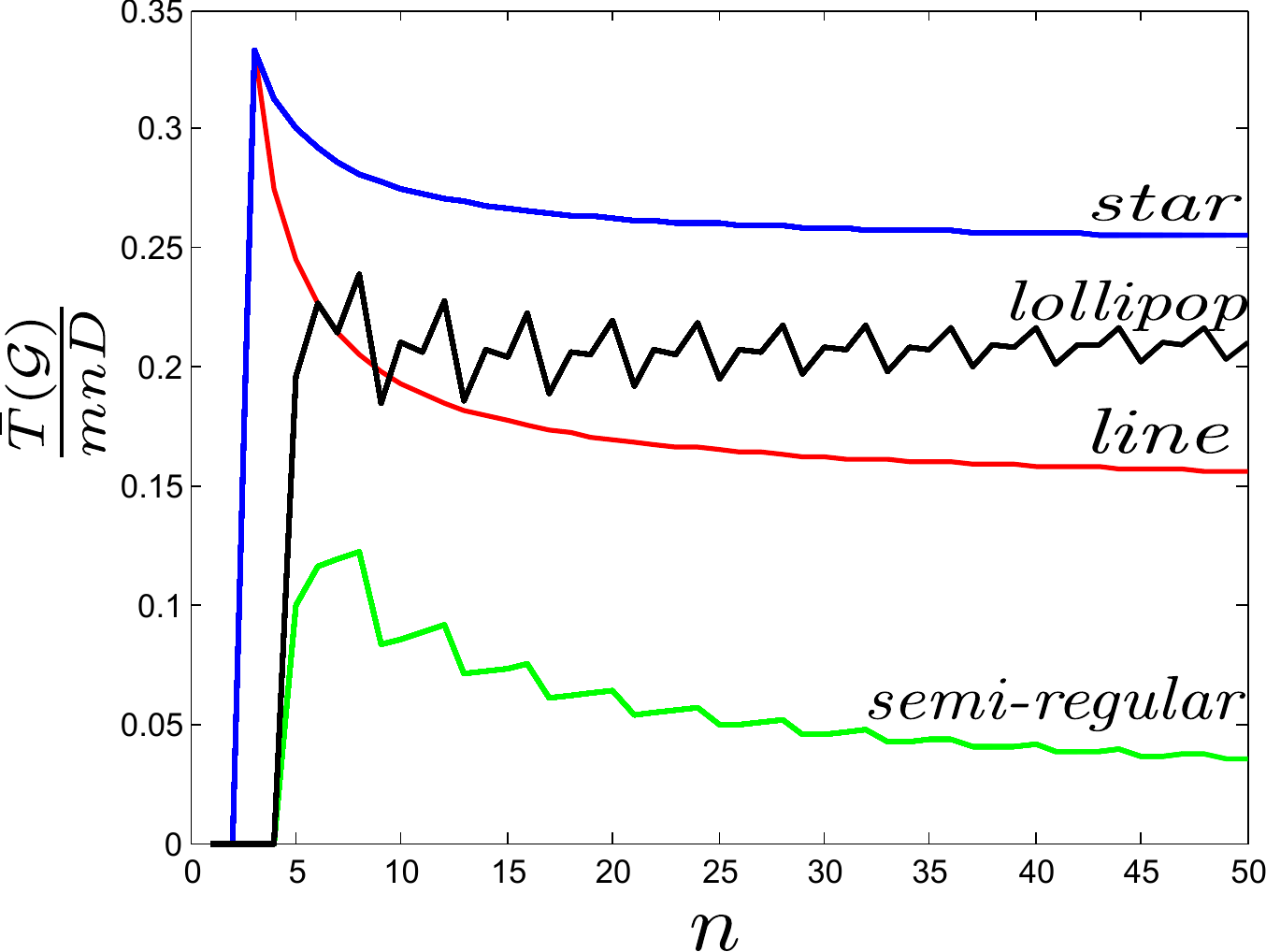} 
\hspace{0.4in}\end{center}
\caption{Comparison between the maximum meeting time and $mnD$ four different types of graphs.\label{fig:all-graphs}}
\end{figure}

\section{Convergence Time over Time-Varying Networks}\label{sec:Time-Varying}

In this section we extend the results of the previous section to time-varying networks. Let us consider a sequence of connected networks $\mathcal{G}(t)=(V,\mathcal{E}(t)), t=0,1,2,\ldots$, over a set of $n$ vertices, $V=\{1,2,\ldots,n\}$. We assume that such a sequence is selected a priori arbitrarily, and then fixed at the beginning of the process. In particular, the future sequences of the graphs do not depend on the current states of the quantized consensus protocol which we will define later. We denote the number of edges and the degree of vertex $x$ in $\mathcal{G}(t)$ by $m_t$ and $d_{t}(x)$, respectively. Here, we assume that all the networks are connected, as otherwise we may not reach consensus through the dynamics. For example if the networks are allowed to be disconnected, one can easily find a sequence of networks such that the quantized algorithm does not converge to any specific outcome. Moreover, as in the case of static networks, and in order to be able to define a virtual process (Remark \ref{rem:star-graph}), we invoke the following assumption.
\begin{assumption}\label{ass:time-varying-graphs}
$\mathcal{G}(t)=(V,\mathcal{E}(t))$ are connected graphs such that $d_{t}(u)+d_t(v)\leq m_t, \forall (u,v)\in \mathcal{E}(t), \forall t=0,1,2,\ldots$.   
\end{assumption}

\smallskip
Note that the above assumption is a very mild one and simply states that the sequence of graphs can be arbitrarily chosen from the set of all connected graphs except star and double-star graphs. Under this assumption we analyze the expected convergence time of the dynamics. 

\smallskip
\begin{remark}
One can relax the connectivity condition by assuming the existence of a constant $B\in \mathbb{N}$ such that the undirected graph with the vertex set $V$ and the edge set $\cup_{k=1}^{B}\mathcal{E}(t+k)$ is connected for all $t\ge 0$. We refer to such sequences of networks as $B$-connected networks. 
\end{remark}

\smallskip
Let us define a sequence of lazy random walks $\{\mathcal{Z}_t\}_{t\ge 0}$ corresponding to each network  with the following transition probability matrices:

\begin{align}\label{eq:transition-time-varying-z}
P_{t}(x,y) = \begin{cases} 1-\frac{d_{t}(x)}{m_t}, & \mbox{if} \ y=x \\ 
\frac{1}{m_t}, & \mbox{if} \  y \in N_t(x) \\
0, & \mbox{else}, 
\end{cases}
\end{align}

where $N_t(x)$ denotes the set of neighbors of node $x$ in $\mathcal{G}(t)$. Note that since all the probability matrices $P_{t}, t\ge 0$ are doubly stochastic, any arbitrary product is also doubly stochastic, and hence they all share a common stationary distribution $\pi=(\frac{1}{n},\frac{1}{n},\ldots,\frac{1}{n})$. Now based on the above setting and very naturally, one can extend the unbiased quantized protocol over the sequences of time-varying graphs. To be more precise, we consider a similar dynamics over these sequences of networks as follows. At each time step $t=0,1,2, \ldots$, we choose an edge uniformly at random from $\mathcal{G}(t)$ and update its incident nodes as in the static case. We note that since such dynamics preserve the average of the opinions over time-varying networks, the function given in \eqref{eq:Lyapunov} is a valid Lyapunov function for the dynamics over time-varying networks. In particular, the value of this function will decrease by at least 2 after every nontrivial update. Thus, as in the static network analysis, the total number of nontrivial averagings is at most of the order of $\mathcal{O}(n)$. Consequently, finding an upper bound on the maximum expected convergence time of the dynamics over time-varying networks reduces to that of finding the maximum expected time it takes for a nontrivial update to take place, which we denote it by $\bar{T}$.

Next, we consider two random walkers who move jointly over these sequences of time-varying networks based on whether the selected edge is incident to them or not. In other words, if the chosen edge at some time instant is incident to one of the walkers, we will then move it, otherwise the walkers will not change their positions. Through some manipulation as in the case of static network analysis, one can argue that $\bar{T}$ is equal to the maximum expected time it takes for these two walkers to meet. Therefore, very similar to the static case, we have the following definitions.
\smallskip

\begin{definition}
Denoting the locations of the walkers at time instant $t$ by $x,y \in V(\mathcal{G}(t))$, if the selected edge at this time is incident to one of the walkers, e. g., $\{x,x_i\}$ for some $x_i\in N_t(x)$, we will move it from node $x$ to node $x_i$, otherwise the walkers will not change their positions. We refer to such a random walks over the sequence of $\mathcal{G}(t), t=0,1,2,\ldots$ as the \textit{original} process.
\end{definition}  
\smallskip
\begin{definition}
In a time-varying network, a virtual process is the same as the original process until when the walkers are each other's neighbors at some time instant $t$, i.e. $x\in N_t(y)$, for some $x,y\in V(\mathcal{G}(t))$. At this time the connecting edge $\{x,y\}\in \mathcal{E}(t)$ in the virtual process is counted twice in the edge probability distribution, i.e., for $e\in \mathcal{E}(t)$,
\begin{align}\nonumber
\!\!\mathbb{P}(e) \!=\! \begin{cases} \ \frac{2}{m_t},\!\!\!\! & \resizebox{0.175\hsize}{!}{$\mbox{if} \ \mbox{$e=\{x,y\}$}$} \\ 
\frac{1}{m_t},\!\!\!\! & \resizebox{0.5\hsize}{!}{$\mbox{if} \ \mbox{$e$ is incident to either $x$ or $y$ on $\mathcal{G}(t)$}$}   \\
\!\!\frac{m_t-d_t(x)-d_t(y)}{m_t(m_t+1-d_t(x)-d_t(y))},\!\!\!\! & \resizebox{0.5\hsize}{!}{$\mbox{if} \ \mbox{$e$ is incident to neither $x$ nor $y$ on $\mathcal{G}(t)$}$}, 
\end{cases}
\end{align}  
and the walkers move depending on whether the selected edge is incident them or not.
\end{definition}

\smallskip
\begin{remark}
Due to the time-varying nature of the networks, there is no dependency between the location of the walkers and the next graph in the sequence. In other words, the next graph in the sequence cannot be determined based on the current locations of the walkers, as otherwise, one can simply construct a sequence of connected time-varying networks which depends on the location of the walkers, such that the walkers never meet each other. We will see later that due to some laziness that exists in the joint transition probability of these random walks, the expected time until they meet is finite.         
\end{remark}

\smallskip
Now let us focus on the virtual process; we denote its transition probabilities described on the network $\mathcal{G}(t) \times \mathcal{G}(t)$ by a matrix $K(t)$. In fact, $K(t)$ is an $n^2\times n^2$ dimensional matrix whose rows and columns are labeled as all the possible pairs of vertices, and the entry $((x_1,y_1),(x_2,y_2))$ of $K(t)$ is the conditional probability that the walkers in the virtual process are at the nodes $(x_2,y_2)$ given that they were at $(x_1,y_1)$ at the previous time step. Based on this construction, the meeting time of the virtual process that started at $(a,b)$ is equal to the expected time that a time inhomogeneous Markov chain with transition matrices $K(t), t=0,1,2,\ldots$ started from $(a,b)$ hits one of the states $S=\{(1,1), (2,2),\ldots, (n,n)\}$ for the first time. In fact, since we are interested in an upper bound on the expected hitting time of such a random walk on $\{\mathcal{G}(t)\times\mathcal{G}(t)\}_{t\ge 0}$ (and hence an upper bound on the expected meeting time of the virtual process), we can manipulate some of the entries of the matrices $\{K(t)\}_{t\ge 0}$ as long as we make sure that the expected hitting time to $S$ does not decrease throughout the process. Therefore, we upper bound the expected hitting time of such a modified process whose transition matrix at time $t$ we denote by $\bar{K}(t)$. Following this idea and in order to have a symmetric modified chain, for all $t\ge 0$ we define $K(t)$ to be the same as $\bar{K}(t)$ in all but the following cases:  
\begin{itemize}
\item We note that in the virtual process when two walkers are each other's neighbors ($x\in N_{t}(y)$), the probability that the connecting edge between them is chosen is $\frac{2}{m_t}$. Therefore, in the modified chain matrix $\bar{K}(t)$ by assigning probabilities $\frac{1}{m_t}$ for moving from $(x,y)$ to $(x,x)$ and also moving from $(x,y)$ to $(y,y)$ (similarly moving from $(y,x)$ to $(x,x)$ and also moving from $(y,x)$ to $(y,y)$), the expected hitting time to $S$ will not change.
\item Since all the vertices $(x,x)\in S$ are absorbing states in the virtual process, we have $K(t)((x,x),(x,x))=1, \forall x\in V$. In this case, by modifying the row $(x,x)$ of the matrix $K(t)$ to $\bar{K}(t)((x,x),(x,x))=1-\frac{2d_t(x)}{m_t}$ and $\bar{K}(t)((x,x),(x',y'))=\frac{1}{m_t}, \forall (x',y')\in N_{\mathcal{G}(t)\times \mathcal{G}(t)}(x,x)$ we will get a chain whose expected hitting time to $S$ is again the same as the expected hitting time of the chain $\{K(t)\}_{t\ge 0}$.
\end{itemize}

\smallskip
By these modifications, the modified chains $\{\bar{K}(t)\}_{t\ge 0}$ and $\{K(t)\}_{t\ge 0}$ will have the same expected hitting times to $S$. Moreover, by the definition of the transition probabilities of the virtual process matrices ($K(t)$) and the above modifications, we observe that the transition matrix $\bar{K}(t)$ must satisfy $\bar{K}(t)=I-\frac{1}{m_t}\mathcal{L}_{\mathcal{G}(t)\times \mathcal{G}(t)}$ for all $t\ge 0$, where $\mathcal{L}_{\mathcal{G}(t)\times \mathcal{G}(t)}$ is the Laplacian of the Cartesian product graph $\mathcal{G}(t)\times \mathcal{G}(t)$ and $I$ denotes the identity matrix of proper size. On the other side by a close look at the matrix $P_t$ it is not hard to see that $P_t=I-\frac{1}{m_t}\mathcal{L}_{\mathcal{G}(t)}$. 

\smallskip
We are now in a position to study the expected convergence time of the dynamics. But, before we proceed, we first provide a summary of the steps involved in the proof. Based on the above discussion, in order to determine the expected meeting time function of the virtual process over $\{\mathcal{G}(t)\}_{t=0}^{\infty}$, we can equivalently concentrate on finding the expected hitting time to the absorbing states $S$ of an inhomogeneous Markov chain with transition matrices $\{K(t)\}_{t=0}^{\infty}$ which are defined over $\{\mathcal{G}(t)\times \mathcal{G}(t)\}_{t=0}^{\infty}$. As discussed above, the hitting time to the absorbing states of this chain is equal to that in the modified chain $\{\bar{K}(t)\}_{t=0}^{\infty}$. Since the matrices $\{\bar{K}(t)\}_{t=0}^{\infty}$ are symmetric, and hence, doubly stochastic, we can find a precise expression for the second largest eigenvalue and the smallest eigenvalue of $\bar{K}(t)$ based on those of the matrix $P_t$. This allows us to find tight bounds on the spectral gap of the matrices $\bar{K}(t), t=0,1,\ldots$.

\smallskip 
Since all matrices in the inhomogeneous Markov chain $\{\bar{K}(t)\}_{t=0}^{\infty}$ are doubly stochastic, starting from any initial distribution  $p(0)$, and after sufficiently long period of time $t'$ (which will be determined by the spectral gap of such matrices), the probability of being in different states $p(t')$ will be very close to the stationary distribution of the chain, i.e., $\pi=(\frac{1}{n^2}, \frac{1}{n^2},\ldots, \frac{1}{n^2})'$. In particular, the probability of being absorbed by $S$ after time $t'$ will be large enough and bounded away from 0. This allows us to find an upper bound on how long it takes until the chain $\{\bar{K}(t)\}_{t=0}^{\infty}$ starting from an arbitrary initial distribution $p(0)$ to hit at least one of the absorbing states. Equivalently, this provides an upper bound on the expected hitting time of the chain $\{K(t)\}_{t=0}^{\infty}$ to the set $S$, and hence, an upper bound on the expected meeting time of the virtual process. Finally, by the same line of argument as in the case of static networks, we can show that the maximum expected meeting time of the original process is within a constant factor of that in the virtual process. Keeping these main steps in mind and toward a complete proof, we first consider the following lemma.

\smallskip    
\begin{lemma}[Laplacian Spectrum of a Graph Product \cite{anderson1985eigenvalues}]\label{lemma:product-graph-eigenvalue}
If $\mathcal{L}_G$ has eigenvalues $\lambda_1,\ldots, \lambda_n$ and
$\mathcal{L}_H$ has eigenvalues $\mu_1,\ldots, \mu_n$, then $\mathcal{L}_{G\times H}$ has 
eigenvalues $\lambda_i+\mu_j, i,j=1,\ldots,n$.
\end{lemma}

\smallskip
From the Perron-Frobenius theorem and Lemma \ref{lemma:product-graph-eigenvalue} one can conclude that the second smallest eigenvalues of $\mathcal{L}_{\mathcal{G}(t)}$ and $\mathcal{L}_{\mathcal{G}(t)\times \mathcal{G}(t)}$ must be the same for $t\ge 0$. Now let us denote the second largest eigenvalue and the second smallest eigenvalue of a $k\times k$ matrix $A$ by $\alpha_2(A)$ and $\alpha_{k-1}(A)$, respectively. For every $t\ge 0$, we can write
\begin{align}\label{eq:eigenvalue-equality}
\alpha_2(\bar{K}(t))&\!=\!\alpha_2(I\!-\!\frac{1}{m_t}\mathcal{L}_{\mathcal{G}(t)\times \mathcal{G}(t)})\!=\!1\!-\!\frac{1}{m_t}\alpha_{n^2-1}(\mathcal{L}_{\mathcal{G}(t)\times \mathcal{G}(t)})\cr
&\!=\!1\!-\!\frac{1}{m_t}\alpha_{n^2-1}(\mathcal{L}_{\mathcal{G}(t)})\!=\!\alpha_2(I\!-\!\frac{1}{m_t}\mathcal{L}_{ \mathcal{G}(t)})\!=\!\alpha_2(P_t).
\end{align}

Similarly, from Perron-Frobenius theorem and Lemma \ref{lemma:product-graph-eigenvalue} one can observe that the largest eigenvalue of $\mathcal{L}_{\mathcal{G}(t)\times \mathcal{G}(t)}$ is two times of that in $\mathcal{L}_{\mathcal{G}(t)}$ for $t\ge 0$. Thus, if we denote, respectively, the largest eigenvalue and the smallest eigenvalue of a $k\times k$ matrix $A$ by $\alpha_1(A)$ and $\alpha_{k}(A)$, we can write
\begin{align}\label{eq:smallest-eigenvalue-bound}
\alpha_{n^2}(\bar{K}(t))&\!=\!\alpha_{n^2}(I\!-\!\frac{1}{m_t}\mathcal{L}_{\mathcal{G}(t)\times \mathcal{G}(t)})\!=\!1\!-\!\frac{1}{m_t}\alpha_{1}(\mathcal{L}_{\mathcal{G}(t)\times \mathcal{G}(t)})\cr
&\!=\!1\!-\!\frac{2}{m_t}\alpha_{1}(\mathcal{L}_{\mathcal{G}(t)}).
\end{align}
In what follows, our goal is to find an upper bound for the second largest eigenvalue and a lower bound for the smallest eigenvalue of the matrix $\bar{K}(t)$. 

First, we note that by relation \eqref{eq:eigenvalue-equality}, the second largest eigenvalue of the matrix $\bar{K}(t)$ is equal to the second largest eigenvalue of the matrix $P_t$. In order to bound the second largest eigenvalue of the matrix $P_t$ we look for a relationship between its eigenvalues and the hitting times of a random walk with transition probability matrix $P_t$. In fact, the random target Lemma provides us with such a relationship.

\smallskip
\begin{lemma}\label{lemma:random-target}
[Random Target Lemma] For an irreducible Markov chain with state space $\Omega=\{1,2,\ldots,n\}$, transition matrix $P$, and stationary distribution $\pi$, we have
\begin{align}\nonumber
\sum_{j=1}^{n}\pi_j H(i,j)=\sum_{k=2}^{n}\frac{1}{1-\alpha_k(P)}, \ \ \forall i\in \Omega
\end{align}
where $1=\alpha_1(P)> \alpha_2(P)\ge \alpha_3(P)\ge \ldots, \alpha_n(P)$ denote the eigenvalues of $P$ in a non-increasing order and $H(\cdot,\cdot)$ is the expected hitting time function of the chain.
\end{lemma} 
\begin{proof}
The idea of the proof is to find a recursion matrix equality for the expected hitting time of a chain with transition probability matrix $P$ and interpreting the solution of this equation based on the eigenvalues of $P$. A complete proof can be found in \cite{lovasz1993random} and also \cite{catral2010kemeny}.  
\end{proof}

\smallskip
Next, in order to obtain a lower bound for the smallest eigenvalue of the matrix $\bar{K}(t)$ and in view of relation \eqref{eq:smallest-eigenvalue-bound}, we find an upper bound for the largest eigenvalue of the Laplacian of $\mathcal{G}(t)$, i.e., $\alpha_{1}(\mathcal{L}_{\mathcal{G}(t)})$. In fact, the following lemma provides us with a desired upper bound.

\smallskip
\begin{lemma}\label{lemm:largest-laplacian-eigenvalue}
The largest eigenvalue of the Laplacian of any graph with $m$ edges, satisfying Assumption \ref{ass:time-varying-graphs}, is bounded from above by $m-\frac{1}{2}$. 
\end{lemma}
\begin{proof}
The proof can be found in Appendix \ref{ap-static-preliminary-proofs}.   
\end{proof}

Finally, based on Lemma \ref{lemma:random-target} and Lemma \ref{lemm:largest-laplacian-eigenvalue} we can state the main result of this section, which is an upper bound for the expected convergence time of quantized consensus over time-varying networks.
\begin{theorem}\label{thm:main-time-varying}
Let $m_{max}=\max_{t\ge 0}m_t$ and $D_{max}=\max_{t\ge 0}D_t$. Then, the expected convergence time of unbiased quantized consensus over time-varying graphs satisfying Assumption \ref{ass:time-varying-graphs} is bounded from above by $\mathcal{O}\big(n^2m_{max}D_{max}\ln^2(n)\big)$.
\end{theorem}
\begin{proof}
Since, for all $t\ge 0$, the matrix $P_t$ is doubly stochastic, $(\frac{1}{n},\frac{1}{n},\ldots,\frac{1}{n})$ is its stationary distribution. Using Lemma \ref{lemma:random-target}, we can write
\begin{align}\nonumber
\frac{1}{1-\alpha_2(P_t)}&\leq \sum_{k=2}^{n}\frac{1}{1-\alpha_k(P_t)}=\frac{1}{n}\sum_{j\neq i} H_{\mathcal{Z}_t}(i,j)\cr 
&\leq \frac{1}{n}(n\times 2nm_tD_t)=2nm_tD_t,
\end{align} 
where the last inequality is due to the relation \eqref{eq:upper-max-hitting-time} for the maximum hitting time that we computed in the case of the fixed graph. Therefore, we get 
\begin{align}\label{eq:alpha_2}
\alpha_2(P_t)\leq 1-\frac{1}{2nm_tD_t}, \forall t\ge 0.
\end{align}
Moreover, using relation \eqref{eq:smallest-eigenvalue-bound} and Lemma \ref{lemm:largest-laplacian-eigenvalue}, we can write, 
\begin{align}\label{eq:alpha_n}
\alpha_{n^2}(\bar{K}(t))&=1-\frac{2}{m_t}\alpha_1(\mathcal{L}_{\mathcal{G}(t)})\ge 1-\frac{2}{m_t}(m_t-\frac{1}{2})\cr 
&=-1+\frac{1}{m_t}, \forall t\ge 0.
\end{align} 
Let the vector $p(t)=\big(p_{(1,1)}(t),p_{(1,2)}(t),\ldots,p_{(n,n)}(t)\big)'$ denote the probability at time $t$ of being at different states of a random walk with transition matrix $\bar{K}(t)$. Since $\bar{K}(t)$ is a doubly stochastic matrix, $\pi:=\pi(t)=(\frac{1}{n^2},\frac{1}{n^2},\ldots,\frac{1}{n^2}), \forall t\ge 0$ is its stationary distribution and the average is preserved throughout the dynamics. Now, we note that since $\bar{K}(t)$ is a real-valued and symmetric matrix, it has an orthogonal set of eigenvectors $\bold{1},v_2,\ldots,v_{n^2}$, corresponding to the eigenvalues $1>\alpha_2(\bar{K}(t))\ge \ldots\ge \alpha_{n^2}(\bar{K}(t))$. Since $(p(t)-\pi)'\bold{1}=0$, $p(t)-\pi$ can be written as $\sum_{k=2}^{n^2}r_kv_k$ for some coefficients $r_k, k=2,\ldots, n^2$. In particular, since $v_k, k=2,\ldots, n^2$, are orthogonal, we have $\|p(t)-\pi\|_2^2=\sum_{k=2}^{n^2}r_k^2$. Now we can write
\begin{align}\nonumber
\left\|\bar{K}(t)p(t)\!-\!\pi\right\|_2^2&\!=\!\left\|\bar{K}(t)(p(t)\!-\!\pi)\right\|_2^2\!=\!\left\|\sum_{k=2}^{n^2}r_k(\bar{K}(t)v_k)\right\|_2^2\cr 
&\!=\!\left\|\sum_{k=2}^{n^2}r_k\alpha_k(\bar{K}(t))v_k\right\|_2^2\!=\!\sum_{k=2}^{n^2}r_k^2\alpha_k^2(\bar{K}(t))\cr &\!\leq\!\max_{k=2,\dots,n^2}\left\{\alpha_k^2(\bar{K}(t))\right\}\sum_{k=2}^{n}r_k^2\cr 
&\!=\!\max\left\{\alpha_2^2(\bar{K}(t)),\alpha_{n^2}^2(\bar{K}(t))\right\}\left\|p(t)-\pi\right\|_2^2.
\end{align}  
Therefore, for every probability vector $p(t)$, we have
\begin{align}\label{eq:fundumental-Markov-convergence}
\left\|\bar{K}(t)p(t)\!-\!\pi\right\|_2&\!\leq\! \max\left\{|\alpha_2(\bar{K}(t))|,|\alpha_{n^2}(\bar{K}(t))|\right\}\left\|p(t)\!-\!\pi\right\|_2\cr &=\max\left\{|\alpha_2(P_t)|,|\alpha_{n^2}(\bar{K}(t))|\right\}\left\|p(t)\!-\!\pi\right\|_2, 
\end{align}
where the equality is due to the relation \eqref{eq:eigenvalue-equality}. Using relations \eqref{eq:alpha_2} and \eqref{eq:alpha_n} in \eqref{eq:fundumental-Markov-convergence}, we get,
\begin{align}\label{eq:decrease-rate-markov-chain}
\left\|\bar{K}(t)p(t)\!-\!\pi\right\|_2&\!\leq\! \max\left\{1\!-\!\frac{1}{2nm_tD_t},1\!-\!\frac{1}{m_t}\right\}\left\|p(t)\!-\!\pi\right\|_2\cr 
&\!=\!\left(1\!-\!\frac{1}{2nm_tD_t}\right)\left\|p(t)\!-\!\pi\right\|_2.
\end{align} 
Since the above argument works for every time instant $t\ge 0$, and for each of the transition matrices $P_t$ and $\bar{K}(t)$, using relation \eqref{eq:decrease-rate-markov-chain} recursively we get,
\begin{align}\nonumber
\left\|p(t)-\pi\right\|_2&=\left\|\bar{K}(t-1)\bar{K}(t-2)\ldots\bar{K}(0)p(0)-\pi\right\|_2\cr 
&\leq\prod_{k=0}^{t-1}\left(1-\frac{1}{2nm_kD_k}\right)\left\|p(0)-\pi\right\|_2\cr
&\leq\left(1-\frac{1}{2nm_{max}D_{max}}\right)^t\left\|p(0)-\pi\right\|_2\cr &=\left(1-\frac{1}{2nm_{max}D_{max}}\right)^t\sqrt{\frac{n^2-1}{n^2}}\cr
&\leq e^{\frac{-t}{2nm_{max}D_{max}}}\sqrt{\frac{n^2-1}{n^2}}<e^{\frac{-t}{2nm_{max}D_{max}}}, 
\end{align}
where $p(0)$ denotes the initial probability, which is 1 in one entry and zero everywhere else. Therefore after at most $4nm_{max}D_{max}(1+\ln(n))$ steps, we get $\big(\frac{1}{2}\big)^{\frac{t}{2nm_{max}D_{max}}}\leq \frac{1}{2n^2}$ and this means $\|p(t)-\pi\|_2\leq \frac{1}{2n^2}$. In other words, for all $t\ge 4nm_{max}D_{max}(1+\ln(n))$, we must have $p_{(i,j)}(t)\in [\frac{1}{2n^2}, \frac{3}{2n^2}], \forall i,j$. In particular, $\sum_{i=1}^{n}p_{(i,i)}(t)\ge n\times \frac{1}{2n^2}=\frac{1}{2n}$. This means that after at most $4nm_{max}D_{max}(1+\ln(n))$ steps the probability of hitting at least one of the states in $S=\{(i,i), i=1\ldots,n\}$ is larger than or equal to $\frac{1}{2n}$. Now, by applying Proposition 4.1 in \cite{zhu2011convergence}, we conclude that the expected hitting time of a random walk with transition probabilities ${\bar{K}(t)}_{t\ge 0}$ is less than or equal to $4n\times 4nm_{max}D_{max}(1+\ln(n))$. This result can be also viewed as a corollary of Lemma 13 in \cite{Explore-fast-graph}. Therefore, since $\{\bar{K}(t)\}_{t\ge 0}$ and $\{K(t)\}_{t\ge 0}$ have the same expected hitting times to the states of $S$, we conclude that the expected meeting time of the virtual process started from any time step $t\ge 0$, which we denote by $M_t$, is bounded from above by $M_t\leq 16n^2m_{max}D_{max}(1+\ln(n))$.

Now, as in the case of static graphs, we argue that the virtual process and the original process are the same until the two walkers are each other's neighbors, i.e. for some $x,y$ and $t\ge 0$, with $x\in N_{t}(y)$. At this time, the probability that two walkers in the original process meet each other at the next time step is at least half of that in the virtual process. In other words, more than half of the times when two walkers are each other's neighbors and they meet in the virtual process, they will meet in the original process as well. Since each of these intersections may happen independently, we can write
\begin{align}\label{eq:original-meeting-time-varying}
\bar{T}&\leq \sum_{k=0}^{\infty}(\frac{1}{2})^k \max_{t_1<t_2<\ldots<t_k}(M_{t_1}+M_{t_2}+\ldots+M_{t_k})\cr 
&\leq \sum_{k=0}^{\infty}(\frac{1}{2})^k k\times 16n^2m_{max}D_{max}(1+\ln(n))\cr
&=32n^2m_{max}D_{max}(1+\ln(n)). 
\end{align}

Finally, since the dynamics preserve the average of the opinions over time varying networks, the function given in \eqref{eq:Lyapunov} is a valid Lyapunov function through the trajectory of the dynamics over time-varying networks, and its value will decrease by at least 2 after every nontrivial update. Thus as in the static network analysis, the total number of nontrivial averagings is at most of the order of $\mathcal{O}(n)$. This, in view of the relation \eqref{eq:original-meeting-time-varying} and Corollary 4 in \cite{shang2013upper}, completes the proof. 
\end{proof}

In fact, Theorem \ref{thm:main-time-varying} improves significantly some of the existing upper bounds for the expected convergence time of quantized consensus over time-varying networks \cite{zhu2011convergence}. Finally, we would like to emphasize that as in the analysis of static networks and using direct analysis for the two special cases of star graph and double-star graph (Lemma \ref{lemm:special-case-star}), one may be able to generalize Theorem \ref{thm:main-time-varying} to such graphs.   

\smallskip
\section{Conclusion}\label{sec:conclusion}
In this paper, we have studied the unbiased quantized consensus problem under the assumption that the underlying network $\mathcal{G}$ is connected. We provided tight upper and lower bounds for the maximum expected convergence time of the model. Further, we provided an exact asymptotic value for the convergence time when the network is a line graph or a cycle. We observed that the given bounds for static networks agree with the simulation results for some particular choices of undirected connected networks. Finally, we extended our results to time-varying networks under the assumption of connectivity over the sequence of networks.  
As a future direction of research, an interesting problem is to consider the model when the choice of edges at each time instant is based on some specific, not necessarily uniform, distribution. Also, given a network $\mathcal{G}$, one can think of adding an edge (or removing an edge) so as to minimize (or maximize) the expected convergence time.

\appendices

\input{apx1}

\input{apx2}

\bibliographystyle{IEEEtran}
\bibliography{thesisrefs}
\end{document}

%% file: apx1.tex
\bigskip
\section{Preliminary results and omitted proofs}\label{ap-static-preliminary-proofs}

\smallskip
In this section, we discuss some preliminary results which will be used to prove our main results in section \ref{sec:mainresults}. We first state some relevant results from the theory of Markov chains. A simple random walk on a graph $\mathcal{G}$ is a Markov chain with transition probabilities 
\begin{align}\nonumber
P(x,y)=\begin{cases} \frac{1}{d(x)}, & \mbox{if} \ \ y\in N(x) \\ 
0, & \mbox{otherwise}, 
\end{cases}
\end{align}
where $d(x)$ denotes the degree of a node $x$ in the graph $\mathcal{G}$. Note that a simple random walk is a special case of a weighted random walk when the weights of all edges in $\mathcal{G}$ are equal to 1. It is well known that every reversible Markov chain is a weighted random walk on a network. Suppose $P$ is a transition matrix on a finite set $S$ which is reversible with respect to the probability distribution $\pi(\cdot)$. Define conductance on edges by $c(x,y)=\pi(x)P(x,y)$ and $c(x):=\sum_{y:\  y\in N(x)}c(x,y)$. Also, the resistance of each edge $e$ is defined to be the inverse of conductance, i.e. $r(e)=\frac{1}{c(e)}$. 

\smallskip
\begin{lemma}\label{lemm:triangleineq}
$\frac{G_{\tau^a_z}(x)}{d(x)}$ is equal to the induced voltage between $x$ and $z$ , i. e. $V_{xz}$ when we define the terminal voltages to be $V_{zz}=0, V_{az}=\frac{G_{\tau^a_z}(a)}{d(a)}$. Moreover, for all $x$ we have    
\begin{align}\nonumber
\frac{1}{2}\big[\mathcal{R}(a\leftrightarrow z)\!+\!\mathcal{R}(z\leftrightarrow x)\!-\!\mathcal{R}(a\leftrightarrow x)\big]=\frac{G_{\tau^a_z}(x)}{d(x)}=V_{xz}.  
\end{align}
\end{lemma}
\begin{proof}
This is the result of Corollary 3 in \cite{kolmogoroff}.
\end{proof}

By taking summation over the above equality and noting that $\sum_{x}G_{\tau^a_z}(x)$ is equal to the expected hitting time of a simple random walk when it starts from $a$ and hits $z$, we get:
\begin{align}\label{eq:hitting}
H(a,z)=\frac{1}{2}\sum_{x}d(x)\big[\mathcal{R}(a\leftrightarrow z)+\mathcal{R}(z\leftrightarrow x)-\mathcal{R}(a\leftrightarrow x)\big]
\end{align}

\smallskip
\begin{lemma}\label{Lemma:harmonic}
Let $\{X_t\}$ be a Markov chain with an irreducible transition matrix $P$, let $B\subset \Omega$, and $h_B:B\rightarrow R$ be a function defined on $B$. The function $h:\Omega \rightarrow R$ defined by $h(x):=\mathbb{E}[h_B(X_{\tau^x_B})]$ is the unique extension $h(\cdot)$ of $h_B$ such that $h(x)=h_B(x)$ for all $x\in B$ and $h$ is harmonic for $P$ at all $x\in \Omega\setminus B$. 
\end{lemma}
\begin{proof}
The proof can be found in \cite{Markov-Book}.
\end{proof}

\smallskip
\begin{lemma}\label{lemm:special-case-star}
For the star graph and the double-star graph with $m$ edges, the maximum expected meeting time $\bar{T}(\mathcal{G})$ is bounded from above by $O(m^2)$.
\end{lemma}
\begin{proof}
Let us denote the meeting time function of the original process by $\bar{M}(\cdot,\cdot):V\times V\rightarrow \mathbb{R}$. For the case of star graph with $m$ edges, a center node $z$, and two leaves $x$ and $y$, it is not hard to see that $\bar{T}(\mathcal{G})=\bar{M}(x,y)$. Because of the symmetric structure of the star graph and by one step recursive expansion of the meeting time function, we can write
\begin{align}\nonumber
\bar{M}(x,y)&=\frac{2}{m}(1+\bar{M}(x,z))+\frac{m-2}{m}(1+\bar{M}(x,y))\cr
\bar{M}(x,z)&=\frac{1}{m}+\frac{m-1}{m}(1+\bar{M}(x,y)).
\end{align}
Solving these two equations we get, $\bar{T}(\mathcal{G})=\bar{M}(x,y)=\frac{m(m+2)}{2}$. 
For the case of the double-star graph, we use a similar line of argument. In a general form, we consider a double-star graph with center nodes $x_1$ and $y_1$ that share $k$ neighbors, for some $k\ge 0$. Moreover, we assume that $x_1$ and $y_1$ have $i$ and $m+1-i$ neighbors, respectively. Such a graph has been depicted in Figure \ref{fig:double_star}. Again, using the symmetry, one can distinguish between 13 different states for the position of the walkers in such a graph. As an example, denoting the location of the two walkers by $x, y$ and looking at Figure \ref{fig:double_star}, one can observe that when $x\in N(x_1)\setminus N(y_1)$ the relative position of the other walker with respect to $x$ can fit into one of the following cases:
\begin{align}\nonumber
\begin{cases} y\in N(y_1)\setminus N(x_1) \\ 
y=y_1 \\
y\in N(x_1)\cap N(y_1)\\
y=x_1\\
y\in N(x_1)\setminus N(y_1).
\end{cases} 
\end{align}
Note that in order to write recursion expansions of the expected meeting time of the original process, and due to symmetry, only the relative positions of the walkers matters. For example for $x\neq y$, $M(x,y)$ is the same for all pairs of $x,y\in N(x_1)$. Therefore, by recursion expansion of the expected meeting time of the original process for such a graph, we obtain a linear system of equations which for simplicity we write in a matrix form as has been shown in \eqref{eq:matrix-double-star}. Solving this system of equations fully characterizes the expected meeting time of the original process for being in different states of the double-star graph which are upper bounded by $O(\frac{i+m^2}{k})$. Since, $k\leq i\leq m$, for the double-star graph we get $\bar{T}(\mathcal{G})=O(m^2)$. 
\begin{align}\label{eq:matrix-double-star}
 \resizebox{1\hsize}{0.07\vsize}{$
  \left( {\begin{array}{ccccccccccccc}
    -\frac{2}{m} & \frac{1}{m} & \frac{2}{m} & 0 & 0 & 0 & 0 & 0 & 0 & 0 & 0 & 0 & 0\\
    -\frac{k - i + 1}{m} & -\frac{i + 1}{m} & 0 & 0 & \frac{1}{m} & 0 & 0 & \frac{k}{m} & 0 & 0 & 0 & 0 & \frac{1}{m}\\
    -\frac{i + k - m}{m} & 0 & -\frac{m - i + 2}{m} & 0 & 0 & 0 & \frac{k}{m} & 0 & \frac{1}{m} & 0 & 0 & 0 & \frac{1}{m}\\
     0 & 0 & 0 & -\frac{m - i + 2}{m} & 0 & \frac{1}{m} & 0 & -\frac{i + k - m}{m} & 0 & 0 & \frac{k - 1}{m} & 0 & \frac{1}{m}\\
     0 & \frac{1}{m} & 0 & 0 & -\frac{m - i + 1}{m} & 0 & 0 & \frac{k}{m} & 0 & 0 & 0 & -\frac{i + k - m + 1}{m} & 0\\
     0 & 0 & 0 & \frac{1}{m} & 0 & -\frac{i + 1}{m} & -\frac{k - i + 1}{m} & 0 & 0 & 0 & \frac{k - 1}{m} & 0 & \frac{1}{m}\\
     0 & 0 & \frac{1}{m} & 0 & 0 & \frac{1}{m} & -\frac{3}{m} & 0 & \frac{1}{m} & 0 & 0 & 0 & 0 \\
     0 & \frac{1}{m} & 0 & \frac{1}{m} & \frac{1}{m} & 0 & 0 & -\frac{3}{m} & 0 & 0 & 0 & 0 & 0\\
     0 & 0 & \frac{1}{m} & 0 & 0 & 0 & \frac{k}{m} & 0 & -\frac{i}{m} & -\frac{k - i + 2}{m} & 0 & 0 & 0\\
     0 & 0 & 0 & 0 & 0 & 0 & 0 & 0 & \frac{2}{m} & -\frac{2}{m} & 0 & 0 & 0\\
     0 & 0 & 0 & \frac{2}{m} & 0 & \frac{2}{m} & 0 & 0 & 0 & 0 & -\frac{4}{m} & 0 & 0\\
     0 & 0 & 0 & 0 & \frac{2}{m} & 0 & 0 & 0 & 0 & 0 & 0 & -\frac{2}{m} & 0\\
     0 & -\frac{i + k - m}{m} & -\frac{k - i + 1}{m} & \frac{k}{m} & 0 & \frac{k}{m} & 0 & 0 & 0 & 0 & 0 & 0 & -1\\
  \end{array} } \right) \left( {\begin{array}{c} 
\bar{M}(x,y) \\ \bar{M}(x_1,y) \\ \bar{M}(x,y_1) \\ \\ \\ \\ \vdots \\ \vdots \\ \\ \\ \\ \\ \bar{M}(x_1,y_1) \\  \end{array} } \right)=  -\left( {\begin{array}{c} 
1 \\ 1 \\ 1 \\ 1 \\ 1 \\ 1 \\ 1 \\ 1 \\ 1 \\ 1 \\ 1 \\ 1 \\ 1 \\  \end{array} } \right).$}
\end{align}
\end{proof}

\begin{figure}[htb]
\vspace{-4cm}
\begin{center}
\includegraphics[totalheight=.3\textheight,
width=.4\textwidth,viewport=0 0 950 950]{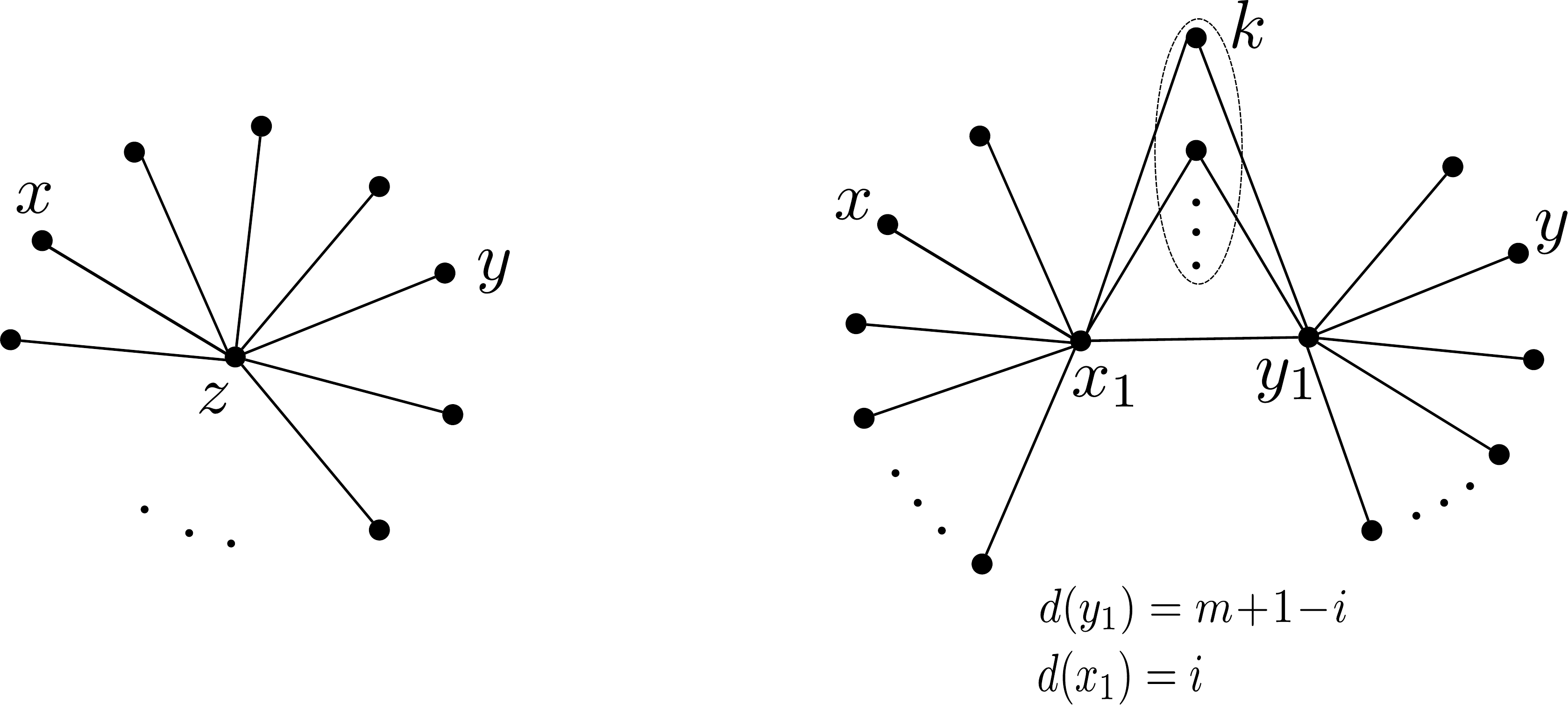} 
\hspace{0.4in}\end{center}
\vspace{-0.5cm}
\caption{Star graph and double-star graph with $m$ edges.\label{fig:double_star}}
\end{figure}

\begin{proof}[proof of Lemma \ref{lemm:largest-laplacian-eigenvalue}]
We use the upper bound of $\alpha_1(\mathcal{L}_{\mathcal{G}})\leq \max\{d(u)+d(v)| (u,v)\in \mathcal{E}(\mathcal{G})\}$ given in \cite{anderson1985eigenvalues}. Since we already assumed $d(u)+d(v)\leq m, \forall (u,v)\in \mathcal{E}(\mathcal{G})$, we consider two cases:
\begin{itemize}
\item $\max\{d(u)+d(v)| (u,v)\in \mathcal{E}(\mathcal{G})\}\leq m-1$. Then, we simply get $\alpha_1(\mathcal{L}_{\mathcal{G}})\leq m-1\leq m-\frac{1}{2}$ 
\item $\max\{d(u)+d(v)| (u,v)\in \mathcal{E}(\mathcal{G})\}=m$. Then, there exists $(u^*, v^*)\in \mathcal{E}(\mathcal{G})$ such that $d(u^*)+d(v^*)=m$. In this case we used the upper bound of $\alpha_1(\mathcal{L}_{\mathcal{G}})\leq 2+ 	
\sqrt{(m -2)(s -2)}$ given in \cite{jiong1997new}, where $s=\max\{d(u)+d(v)|(u, v)\in \mathcal{E}(\mathcal{G})\setminus (u^*,v^*)\}$. Since $d(u^*)+d(v^*)=\max\{d(u)+d(v)| (u,v)\in \mathcal{E}(\mathcal{G})=m$, there exists exactly one edge in $\mathcal{G}\setminus \{u^*,v^*\}$ such that for any $(u, v)\in \mathcal{E}(\mathcal{G})\setminus (u^*,v^*)$ we must have $d(u)+d(v)\leq m-1$. This shows that $s\leq m-1$, and hence,
\begin{align}\nonumber
\alpha_1(\mathcal{L}_{\mathcal{G}})&\!\leq\!2\!+\!\sqrt{(m\!-\!2)(m\!-\!3)}\!=\!2\!+\!\sqrt{(m \!-\!\frac{5}{2})^2\!-\!\frac{1}{4}}\cr 
&\!<\!m\!-\!\frac{1}{2}.
\end{align}  
\end{itemize} 
Therefore, in both cases we have $\alpha_1(\mathcal{L}_{\mathcal{G}})\leq m-\frac{1}{2}$, and this completes the proof.
\end{proof}

%% file: apx2.tex
\bigskip
\section{Tight bounds for simple static networks}\label{ap-static-alternative-proofs}

\smallskip
In this Appendix, we develop an alternative approach in order to study the maximum expected convergence time of unbiased quantized consensus over static graphs. In particular, we identify the precise order of the maximum expected convergence time for the case of simple static graphs such as line graph and cycle. Here we note that, although the result of this appendix works well when we benefit from inherent symmetry in the underlying graph $\mathcal{G}$, in general it does not lead to an explicit tight bound based on the parameters of the network.

\begin{definition}
A \textit{birth-and-death chain} of length $n+1$ has state space $\Omega=\{0, 1, \ldots, n\}$ such that in one step the state can increase or decrease by at most 1.
\end{definition}

\smallskip
\begin{lemma}\label{Lemm:slow-Markov}
Assume that $\mathcal{G}$ is a connected graph with diameter $D$. Then, $\bar{T}(\mathcal{G})$ is bounded from above by the maximum hitting time of a birth-and-death chain of length $D+1$ and positive transition probabilities greater than $\frac{1}{m}$. 
\end{lemma} 
\begin{proof}
We partition all the different states of the above original coupled random walks (Definition \ref{def:original-process-static}) into different classes. Here, we refer to each state as a possible pair of positions of the walkers in the network. For each state $x$ we define $d(x^{(0)}, x^{(2)})$ to be the length of the shortest path between walker 0 and walker 2 in state $x$. Let
\begin{align}\nonumber
S_{\ell}=\{x\in \mathcal{X(\mathcal{G})}|\ \ d(x^{(0)}, x^{(2)})=\ell\}, \ \ell=0, 1, \ldots, D. 
\end{align}
It is clear that $\{S_{\ell}\}_{\ell=1}^{D}$ is a partitioning of all the states. Furthermore, $S_0$ contains just one state. In other words, when we reach class $S_0$, it means that the walkers have met. Now, we introduce a new Markov chain, by letting each class to be a single state by itself, and we denote it by $S_{\ell}$. Finally, we assign the following transition probabilities to the new Markov chain. For each $\ell=1, 2, \ldots, D$, let
\vspace{0.3cm}
\begin{itemize}
\item[1.]  $ \mathbb{P}\{S_{\ell}\rightarrow S_{\ell-1}\}= \underset{x\in S_{\ell}, y\in S_{\ell-1}}{\min} \mathbb{P}\{x\rightarrow y\}, $ \\
\item[2.]  $ \mathbb{P}\{S_{\ell}\rightarrow S_{\ell}\}= \underset{x\in S_{\ell}}{\min}\mathbb{P}\{x\rightarrow x\}, $ \\ 
\item[3.]  $ \mathbb{P}\{S_{\ell}\!\rightarrow \!S_{\ell+1}\}\!=\!1-\underset{x\in S_{\ell}}{\min}\mathbb{P}\{x\!\rightarrow\! x\}-\!\!\!\!\underset{x\in S_{\ell}, y\in S_{\ell-1}}{\min}\!\!\!\!\!\mathbb{P}\{x\rightarrow y\}. $
\end{itemize}
\vspace{0.3cm}
Also, note that $\mathbb{P}\{S_{\ell}\!\rightarrow \!S_{\ell+1}\}\ge 0$, and hence the above transition probabilities are well defined. In fact, assigning the above transition probabilities for the new chain is based on a worst case scenario which keeps the walkers away from each other for the longest period of time. In other words, this new birth-and-death chain slows down the progress of moving the walkers toward each other. As an example, given that the distance of the walkers at the current time instant is $\ell$, i.e., $d(x^{(0)}, x^{(2)})=\ell$, (and hence $x\in S_{\ell}$) the probability that at the next time instant the walkers will get closer to each other in the original process (in the new birth-and-death chain this means that $x$ moves from $S_{\ell}$ to $S_{\ell-1}$) is at least as large as that in the new birth-and-death chain. Therefore, it is not hard to see that the probability that the walkers in the original process meet over every sample path is at least as large as the probability of the equivalent associated sample path in the new birth-and-death process to hit the class $S_0$. Hence, the expected time to hit the state $S_0$ is always an upper bound for meeting time in the original coupled process. Finally, we note that since in the original process each edge is chosen with probability $\frac{1}{m}$, the above assigned probabilities cannot be smaller than $\frac{1}{m}$.   
\end{proof}

\smallskip
\begin{corollary}\label{corr:cycle}
Assume that $\mathcal{G}$ is a cycle of $n$ nodes. Then, 
\begin{align}\nonumber
\bar{T}(\mathcal{G})\leq \frac{n(n-1)(n-3)}{16}+\frac{2n+1}{2}.
\end{align}
\end{corollary}
\vspace{0.35cm}
\begin{proof}
Analyzing the birth-and-death chain described in Lemma \ref{Lemm:slow-Markov} for a cycle with $n$ nodes, we can bound $\bar{T}(\mathcal{G})$ from above. For such a graph, the new Markov chain has the structure shown in Figure \ref{fig:slow-chain}.

\begin{figure}[htb]
\vspace{-3cm}
\begin{center}
\includegraphics[totalheight=.25\textheight,
width=.3\textwidth,viewport=0 0 500 500]{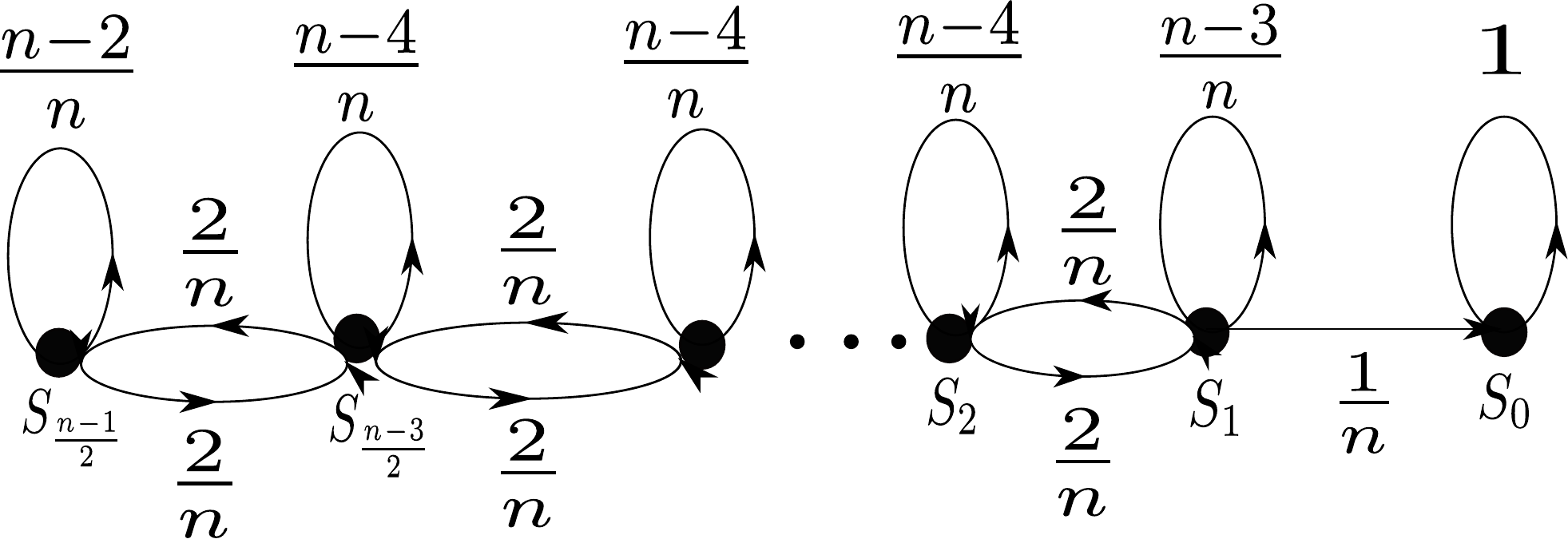} \hspace{0.4in}
\end{center}
\vspace{-0.5cm}
\caption{Birth-and-death chain for a cycle with $n$ nodes.\label{fig:slow-chain}} 
\end{figure}

Therefore, $\bar{T}(\mathcal{G})$ is bounded from above by $H(S_{[\frac{n-1}{2}]},S_0)$ in the birth-and-death diagram of Fig. \ref{fig:slow-chain}. A simple calculation shows that $H(S_{[\frac{n-1}{2}]},S_0)=\frac{n(n-1)(n-3)}{16}+\frac{2n+1}{2}$ and the result follows. 
\end{proof}  
 
\smallskip 
\begin{corollary}\label{corr:path}
If $\mathcal{G}$ is a line graph with $n$ nodes, then
\begin{align}\nonumber
\bar{T}(\mathcal{G})\leq \frac{(n-1)^2(n+1)}{4}.
\end{align}
\end{corollary}
\begin{proof}
This follows from a similar argument as in the proof of Corollary \ref{corr:cycle}, and the bound coincides with the result given in \cite{Kashyab}.  
\end{proof} 

\smallskip
\begin{corollary}\label{corr:equality-line-graph}
For a line graph and cycle with $n$ nodes, we have $C_1n^3\leq\bar{T}(\mathcal{G})\leq C_2n^3$, where $0<C_1<C_2$ are two constants.
\end{corollary}
\begin{proof}
We prove the result for the line graph; for cycle graph the proof is similar. From \eqref{eq:hitting}, since the degree of each node is at most 2, we have 
\begin{align}\nonumber
2\sum_{i}\big[\mathcal{R}\big(x\leftrightarrow y\big)+\mathcal{R}\big(y \leftrightarrow i\big)-\mathcal{R}\big(x\leftrightarrow i\big)\big]\ge H(x,y).
\end{align}
Replacing this inequality in \eqref{eq:partition-lazy}, and since $m=n-1$, we get 
\begin{align}\nonumber
H_{\mathcal{Z}}(x,y)\ge \frac{n+3}{4}H(x,y).
\end{align}
Also, using Theorem \ref{thm:end-main} we can write 
\begin{align}\nonumber
\bar{T}(\mathcal{G})&\ge\frac{1}{2}H_{\mathcal{Z}}\!\ge\!\frac{n\!+\!3}{8}\max_{x,y} H(x,y)=\frac{n\!+\!3}{8}(n\!-\!1)^2. 
\end{align}  
This relation in view of Corollary \ref{corr:path} completes the proof. 
\end{proof}